\documentclass[conference]{IEEEtran}
\usepackage{amsmath,amssymb,cite,algorithm,algorithmic,amsthm}
\usepackage{graphicx}

\begin{document}
\title{Maximal Ratio Transmission\\ in Wireless Poisson Networks\\ under Spatially Correlated Fading Channels}
\author{
\IEEEauthorblockN{George~C.~Alexandropoulos and Marios Kountouris}
\IEEEauthorblockA{Mathematical and Algorithmic Sciences Lab, France Research Center, Huawei Technologies Co$.$ Ltd$.$, \\20 Quai du Point du Jour, 92100 Boulogne-Billancourt, France}
emails: \{george.alexandropoulos, marios.kountouris\}@huawei.com}

\maketitle

\begin{abstract}
The downlink of a wireless network where multi-antenna base stations (BSs) communicate with single-antenna mobile stations (MSs) using maximal ratio transmission (MRT) is considered here. The locations of BSs are modeled by a homogeneous Poisson point process (PPP) and the channel gains between the multiple antennas of each BS and the single antenna of each MS are modeled as spatially arbitrarily correlated Rayleigh random variables. We first present novel closed-form expressions for the distribution of the power of the interference resulting from the coexistence of one intended and one unintended MRT over the considered correlated fading channels. The derived expressions are then used to obtain closed-form expressions for the success probability and area spectral efficiency of the wireless communication network under investigation. Simulation results corroborate the validity of the presented expressions. A key result of this work is that the effect of spatial correlation on the network throughput may be contrasting depending on the density of BSs, the signal-to-interference-plus-noise ratio (SINR) level, and the background noise power. 
\end{abstract}
\thispagestyle{empty}
\IEEEpeerreviewmaketitle

\section{Introduction}\label{sec:Intro}
Multiple-input multiple-output (MIMO) systems have been instrumental in increasing the spectral efficiency and impro\-ving the reliability of wireless links by combating channel fading and potentially reducing interference. Several multi-antenna techniques have been proposed, implemented, and standardized over the past decade, including diversity reception, beamforming, and spatial multiplexing. Diversity-oriented schemes enable multi-antenna transceivers to enhance the system performance by exploiting the diversity provided by the multiple channel fading variations. A classical antenna diversity technique is maximal ratio combining (MRC) \cite{MRC}, where the signals from the multiple receive antennas are weighted such that the signal-to-noise ratio (SNR) of their sum is maximized in the absence of interference, or when interference is treated as background noise. Maximum ratio transmission (MRT) \cite{J:Lo_MRT} is the dual of MRC at the transmitter side, i$.$e$.$ the transmit antenna weights are matched to channel fading, and exhibits the same performance. 

It has been long recognized that the theoretical performance gains of multi-antenna communication are degraded in practice due to \emph{correlated fading} \cite{J:Zhang, J:Lombardo1999, J:Loyka2001, Chuah02,J:Tulino2005, J:Alexandg2009}. Spatially correlated fading channels are usually encountered in multi-antenna systems employing not sufficiently wide separated antennas or with insufficient scattering around the transmitter. While spatial correlation has been considered as a drawback in both single- and multi-user MIMO systems (see e$.$g$.$ \cite{C:Park2005}), correlated fading may have a positive impact on spatial diversity techniques if opportunistic multi-user scheduling is employed \cite{C:Kim2008}.

The performance of diversity-oriented multi-antenna techniques in multi-cell networks has not been extensively studied until very recently, mainly due to significant difficulties in characterizing and modeling the co-channel/out-of-cell interference. Recent advances in spatial network modeling using tools from stochastic geometry were key enablers for analyzing large spatial networks, departing from the classical simplifying assumption of Gaussian interference. Several recent papers have studied multi-antenna processing techniques in wireless networks with Poisson distributed interferers \cite{Hunter08,C:Jindal_2009,OC10,Louie11}, ignoring, however, the effect of spatially correlated fading. In \cite{ Haenggi12, Tanbourgi}, the effect of spatial interference correlation across diversity branches on the performance of single-input multiple-output (SIMO) links in Poisson networks was quantified. These works showed that multi-antenna diversity techniques suffer a diversity loss when spatial interference correlation is properly accounted for. Despite this progress, the performance characte\-rization of multi-antenna techniques in spatially correlated fading channels remains open, and is the main focus of this paper. In contrast to our paper, all aforementioned papers consider spatially uncorrelated Rayleigh fading, i$.$e$.$ channel vectors are spatially white random vectors, and interference correlation occurs due to the common interference locations.

In this work, we investigate the effect of spatially arbitrarily correlated Rayleigh fading on the performance of MRT in wireless Poisson networks, where interferers form a  Poisson point process (PPP) and fading correlation at the antenna branches is present only at the transmitter. The latter consideration describes, for instance, a network set up in which multiple antennas are placed at a high-point base station (BS), while the mobile station (MS) is located in a rich scattering surrounding. We provide analytical expressions for the success probability and area spectral efficiency as a means to capture the effect of spatial fading correlation on the network performance. For that, we present novel closed-form expressions for the distribution of the squared norm of a normalized version of the dot product of two uncorrelated random vectors, each having correlated elements described by a common covariance matrix. The validity of our theoretical analysis is corroborated through comparisons with equivalent computer simulations. 

\textit{Notation:} Vectors and matrices are denoted by boldface lowercase letters and boldface capital letters, respectively. The transpose conjugate is denoted by $(\cdot)^{\rm H}$, while $\mathbf{I}_{n}$ ($n\geq2$) is the $n\times n$ identity matrix. $||\mathbf{a}||$ stands for the Euclidean norm of $\mathbf{a}$, $[\mathbf{a}]_n$ represents $\mathbf{a}$'s $n$th element, $[\mathbf{A}]_{i,j}$ represents $\mathbf{A}$'s $(i,j)$th element, and ${\rm diag}\{\mathbf{a}\}$ denotes a square diagonal matrix with $\mathbf{a}$'s elements in its main diagonal. $\mathcal{N}$, $\mathcal{R}$ and $\mathcal{C}$ represent the natural, real, and complex number sets, respectively, whereas $\mathbb{E}\{\cdot\}$ is the expectation operator and $\mathbb{P}[\cdot]$ represents probability. The Laplace transform (LT) of RV $x$ is denoted as $\mathcal{L}_x\{\cdot\}$ and $x\sim\mathcal{C}\mathcal{N}\left(0,\sigma^{2}\right)$ indicates that $x$ is a circularly-symmetric complex Gaussian RV with zero mean and variance $\sigma^{2}$. Finally, $\Gamma(\cdot)$ is the Gamma function \cite[eq. (8.310/1)]{B:Gra_Ryz_Book}, ${\rm Ei}(\cdot)$ is the exponential integral \cite[eq. (8.211/1)]{B:Gra_Ryz_Book}, and $_2F_1(\cdot,\cdot;\cdot;\cdot)$ denotes the Gauss Hypergeometric function \cite[eq. (9.14)]{B:Gra_Ryz_Book}.

\section{System and Channel Model}\label{sec:System_Model}
We consider the downlink of a wireless network, as shown in Fig.~\ref{Fig:MultiUser_MISO_System}, where BSs are arranged according to a spatial homogeneous PPP $\Phi \subset \mathcal{R}^2$ of density $\lambda$ in the Euclidean plane. Each BS is equipped with $n_{\rm T}\geq2$ transmit antennas and communicates with one single-antenna MS at a fixed distance. Our results can be easily extended for random link distances, e$.$g$.$ assuming a geographically nearest BS connectivity model. 
All simultaneous downlink transmissions are assumed perfectly synchronized, and each BS $k$, with $k\in\mathcal{N}_{+}$, processes individually its complex symbols $s_k$ with a linear precoding vector $\mathbf{f}_k\in\mathcal{C}^{n_{\rm T}}$ before transmitting. It is also assumed that for each value of $k$, $||\mathbf{f}_k||=1$ and $\mathbb{E}\{|s_k|^2\} \leq {\rm P}$, where ${\rm P}$ denotes the transmit power of each BS. 
\begin{figure}[!t]
\centering
\includegraphics[keepaspectratio,width=2.55in]{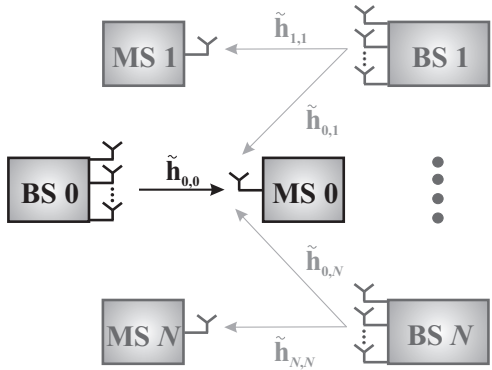}
\caption{The considered downlink wireless network. BS $0$ communicates with the typical MS $0$ under the presence of interference links (only $N$ of the interference links are illustrated in this figure for exposition convenience).}
\label{Fig:MultiUser_MISO_System}
\end{figure}

Let $\mathbf{\tilde{h}}_{k,\ell}\in\mathcal{C}^{n_{\rm T}}$ be the narrowband channel vector for the wireless link between MS $k$ and BS $\ell$ and also let $r_{k,\ell}$ denote the Euclidean distance between these nodes. The baseband received signal at MS $k$ can be mathematically expressed as 
\begin{equation}\label{Eq:Received_Signal}
y_k = r_{k,k}^{-\alpha /2}\mathbf{\tilde{h}}^{\rm H}_{k,k}\mathbf{f}_ks_k + \mathcal{I}_k + n_k
\end{equation}
where the standard power-law pathloss model $l(r)=r^{-\alpha}$ with common pathloss exponent $\alpha>2$ for all links is assumed, $\mathcal{I}_k$ is the cumulative interference from all BSs except the associated BS $k$, which is given by 
\begin{equation}\label{Eq:Interference_k}
\mathcal{I}_k = \sum_{i\in\Phi\backslash\{k\}} r_{k,i}^{-\alpha /2}\mathbf{\tilde{h}}_{k,i}^{\rm H}\mathbf{f}_is_i,
\end{equation}
and $n_k\sim\mathcal{C}\mathcal{N}(0,\sigma^2) $ denotes the additive white Gaussian noise (AWGN). Based on the considered system model described by \eqref{Eq:Received_Signal} and \eqref{Eq:Interference_k}, the signal-to-interference-plus-noise ratio (SINR) at MS $k$ can be computed as
\begin{equation}\label{Eq:SINR_k}
\gamma_k = r_{k,k}^{-\alpha}|\mathbf{\tilde{h}}_{k,k}^{\rm H}\mathbf{f}_k|^2\left(\sum_{i\in\Phi/\{k\}} r_{k,i}^{-\alpha/2}|\mathbf{\tilde{h}}_{k,i}^{\rm H}\mathbf{f}_i|^2+\frac{\sigma^2}{\rm P}\right)^{-1}.
\end{equation}
The elements of each channel gain vector are assumed to be arbitrarily correlated and in particular, each $\mathbf{\tilde{h}}_{k,\ell}$ is modeled as  
$\mathbf{\tilde{h}}_{k,\ell}=\mathbf{R}^{1/2}\mathbf{h}_{k,\ell}$, where $\mathbf{h}_{k,\ell}\in\mathcal{C}^{n_{\rm T}}$ is a standard circularly-symmetric jointly-Gaussian random vector \cite[Sec. 3.7]{B:Gallager} and $\mathbf{R}\triangleq\mathbb{E}\{\mathbf{\tilde{h}}_{k,\ell}\mathbf{\tilde{h}}_{k,\ell}^{\rm H}\}$ denotes the common covariance matrix of $\mathbf{\tilde{h}}_{k,\ell}$'s. In addition, we assume that $\mathbf{R}$ is real symmetric having ones in its main diagonal as well as that $\mathbf{\tilde{h}}_{k,\ell}$'s are independent across the different BSs and of both $r_{k,\ell}$'s and $n_k$'s.  

In this paper, we assume that each BS $k$ knows perfectly $\mathbf{\tilde{h}}_{k,k}$ and treats interference as background noise. BSs do not cooperate \cite{J:Lozano_Limits2013} and each of them acts selfishly aiming at maximizing its own SNR. In order to accomplish that, each BS $k$ is assumed to perform channel-matched precoding, a$.$k$.$a$.$ MRT. Hence, $\mathbf{f}_k$ at each BS $k$ is designed as \cite[Sec. III]{J:Lo_MRT}
\begin{equation}\label{Eq:MRT_vector}
\mathbf{f}_k = \frac{\mathbf{\tilde{h}}_{k,k}}{\|\mathbf{\tilde{h}}_{k,k}\|}.
\end{equation}

\section{Performance Analysis}\label{sec:Analysis}
In this section we present analytical expressions for the success probability and the area spectral efficiency of the conside\-red downlink wireless network. The network performance is ana\-lyzed assuming a typical MS $0$ located at the origin. As a result of Palm probabilities (see \cite{J:Andrews_PPP_Cellular} and references therein for details) and the stationarity of the PPP, the statistics of signal reception at the typical MS $0$ is seen by any MS in the considered wireless network. 

\subsection{Success Probability}
The success probability of a typical MS is defined as 
\begin{equation}\label{Eq:Conditional_Coverage_definition}
P_{\rm suc} = \mathbb{P}\left[\gamma_0>\gamma_{\rm th}\right]
\end{equation}
where $\gamma_{\rm th}$ denotes a predetermined SINR threshold. By substituting \eqref{Eq:MRT_vector} into \eqref{Eq:SINR_k} for MS $0$ and then into \eqref{Eq:Conditional_Coverage_definition}, $P_{\rm suc}$ can be expressed using the cumulative distribution function (CDF) of $\|\mathbf{\tilde{h}}_{0,0}\|^2$ and the expectation of the real RV $\Psi_0$ that represents the aggregate interference, and is given by
\begin{equation}\label{Eq:Psi_0}
\Psi_0 = \sum_{i\in\Phi/\{0\}} r_{0,i}^{-\alpha}\left|g_i\right|^2
\end{equation}
with the complex RV $g_i$ defined as
\begin{equation}\label{Eq:g_variable}
g_i = \frac{\mathbf{\tilde{h}}_{0,i}^{\rm H}\mathbf{\tilde{h}}_{i,i}}{\|\mathbf{\tilde{h}}_{i,i}\|}.
\end{equation}
In particular, $P_{\rm suc}$ can be obtained as 
\begin{equation}\label{Eq:Coverage_2}
P_{\rm suc} = 1-\mathbb{E}_{\Psi_0}\left\{F_{\|\mathbf{\tilde{h}}_{0,0}\|^2}\left[r_{0,0}^{\alpha}\gamma_{\rm th}\left(\Psi_0+\frac{\sigma^2}{\rm P}\right)\right]\right\}
\end{equation}
where $\mathbb{E}_{\Psi_0}\{\cdot\}$ represents the expectation of $\Psi_0$. By setting the Nakagami parameter $m=1$ in \cite[eq. (20)]{J:Zhang} (resulting in the Rayleigh case) and integrating, a closed-form expression for the CDF of $\|\mathbf{\tilde{h}}_{0,0}\|^2$ is easily derived as
\begin{equation}\label{Eq:CDF}
F_{\|\mathbf{\tilde{h}}_{0,0}\|^2}(x) = \sum_{\ell=1}^{n_{\rm T}} \prod_{\substack{i=1\\ i \neq \ell}}^{n_{\rm T}}\left(1-\frac{\left[\boldsymbol{\mu}\right]_i}{\left[\boldsymbol{\mu}\right]_\ell} \right)^{-1}
\left[1-\exp\left(-\frac{x}{\left[\boldsymbol{\mu}\right]_\ell}\right)\right]
\end{equation}
where $\boldsymbol{\mu}\in\mathcal{R}^{n_{\rm T}}$ contains the distinct eigenvalues of $\mathbf{R}$ in decreasing order (the non-diagonal elements of $\mathbf{R}$ lie in $[0,1)$). Hence, plugging \eqref{Eq:CDF} into \eqref{Eq:Coverage_2} and using the definition of the LT of $\Psi_0$, the following expression for $P_{\rm suc}$ can be deduced
\begin{equation}\label{Eq:Coverage_3}
\begin{split}
P_{\rm suc} =& 1-\sum_{\ell=1}^{n_{\rm T}} \prod_{\substack{i=1\\ i \neq \ell}}^{n_{\rm T}}\left(1-\frac{\left[\boldsymbol{\mu}\right]_i}{\left[\boldsymbol{\mu}\right]_\ell} \right)^{-1}
\\&\times\left[1-\exp\left(-\frac{r_{0,0}^{\alpha}\sigma^2\gamma_{\rm th}}{{\rm P}\left[\boldsymbol{\mu}\right]_\ell}\right)\mathcal{L}_{\Psi_0}\left(\frac{r_{0,0}^{\alpha}\gamma_{\rm th}}{\left[\boldsymbol{\mu}\right]_\ell}\right)\right]
\end{split}.
\end{equation}

To obtain a closed-form expression for $P_{\rm suc}$ based on the latter expression \eqref{Eq:Coverage_3}, we first derive an exact as well as an approximate analytical expression for the probability density function (PDF) of $|g_i|^2$ $\forall$ $i\in\mathcal{N}_{+}$ by means of the following two theorems. 
\newtheorem{theorem}{Theorem}
\begin{theorem}\label{Theorem:Exact_Distribution_y}\rm 
The PDF of the RV ${\rm g}\triangleq|g_i|^2$ $\forall$ $i\in\mathcal{N}_{+}$ for the special case where $n_{\rm T}=2$ is given by    
\begin{equation}\label{Eq:Exact_PDF_g}
f_{\rm g}(x) = \frac{1-\rho^2}{2\rho}\left[\mathcal{F}\left(1+\rho,x\right)-\mathcal{F}\left(1-\rho,x\right)\right]
\end{equation}
where $\rho\triangleq\left[\mathbf{R}\right]_{1,2}\in[0,1)$ and function $\mathcal{F}(\eta,x)$, with $\eta>0$ and $x\geq0$, is defined as
\begin{equation}\label{Eq:Function_F}
\begin{split}
\mathcal{F}(\eta,x) =& \frac{\eta}{4(2-\eta)}\exp\left(-\frac{x}{\eta}\right)
\\&-\frac{2-x}{8}\exp\left(-\frac{x}{2}\right){\rm Ei}\left[x\left(\frac{1}{2}-\frac{1}{\eta}\right)\right].
\end{split}
\end{equation}
\end{theorem}
\begin{proof}
The proof is provided in Appendix~\ref{App:Proof_Theorem_1}.
\end{proof}
\begin{theorem}\label{Theorem:Approximate_Distribution_y}\rm
The PDF of ${\rm g}$ for the general case $n_{\rm T}\geq2$ can be approximated by the following exponential PDF  
\begin{equation}\label{Eq:Approximate_PDF_g}
f_{\rm g}(x) = \frac{1}{\sigma_{\rm g}}\exp\left(-\frac{x}{\sigma_{\rm g}}\right)
\end{equation}
where the scale parameter $\sigma_{\rm g}$ is given by
\begin{equation}\label{Eq:sigma_g}
\sigma_{\rm g} = \sum_{n=1}^{n_{\rm T}}\frac{\left[\boldsymbol{\mu}\right]_n^2}{\left[\boldsymbol{\mu}\right]_n}.
\end{equation}
\end{theorem}
\begin{proof}
See Appendix~\ref{App:Proof_Theorem_2}. 
\end{proof}
\textit{Remark 1:} Theorem~\ref{Theorem:Approximate_Distribution_y} also holds for the special case where $\mathbf{R}=\mathbf{I}_n$ (spatially uncorrelated Rayleigh fading), i$.$e$.$ $\left[\boldsymbol{\mu}\right]_n=1$ $\forall$ $n=2,3,\ldots,n_{\rm T}$ (identical eigenvalues). For this special case, \eqref{Eq:sigma_g} simplifies to $\sigma_{\rm g}=1$ and Theorem~\ref{Theorem:Approximate_Distribution_y} provides the exact PDF of ${\rm g}$ for $n_{\rm T}\geq2$, which coincides with the equivalent result in \cite{C:Jindal_2009} and \cite{J:Shah00}. 
\begin{figure}[!t]
\centering
\includegraphics[keepaspectratio,width=3.1in]{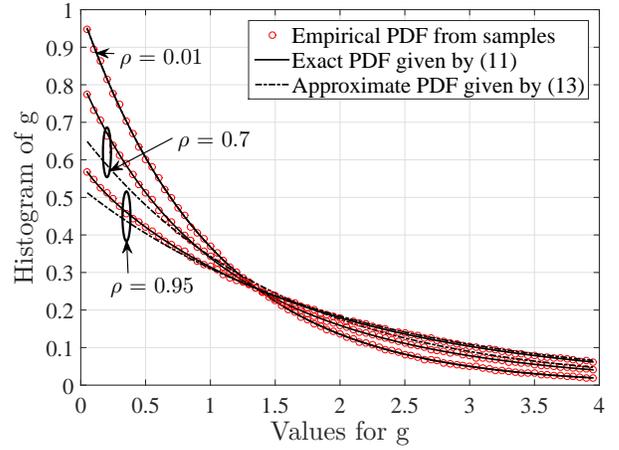}
\caption{Histogram of the positive real RV ${\rm g}$ for the special case of $n_{\rm T}=2$ and various values of the correlation coefficient $\rho$.}
\label{Fig:PDF_nT2}
\end{figure}

\begin{figure}[!t]
\centering
\includegraphics[keepaspectratio,width=3.1in]{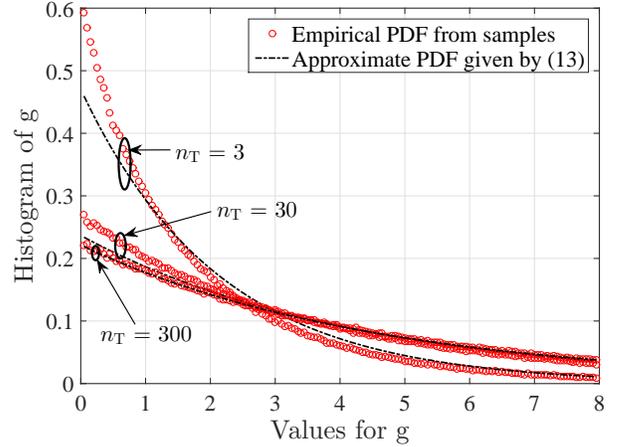}
\caption{Histogram of the positive real RV ${\rm g}$ for a correlation coefficient $\rho=0.8$ and various values of $n_{\rm T}$.}
\label{Fig:PDF_nT}
\end{figure}
The validity of the exact PDF expression in Theorem~\ref{Theorem:Exact_Distribution_y} and the tightness of the approximate PDF expression in Theorem~\ref{Theorem:Approximate_Distribution_y} are demonstrated in Figs$.$~\ref{Fig:PDF_nT2} and \ref{Fig:PDF_nT}, respectively. In these figures the histogram of ${\rm g}$ is plotted for various exponential correlation matrices, i$.$e$.$ $\left[\mathbf{R}\right]_{i,j}=\rho^{|i-j|}$ \cite{J:Loyka2001}, and various values for $n_{\rm T}$. As clearly shown in Fig$.$~\ref{Fig:PDF_nT2}, the curves for the exact PDF expression for $n_{\rm T}=2$ given by \eqref{Eq:Exact_PDF_g} match perfectly with the curves for the equivalent empirical PDF. In addition, as depicted in Fig$.$~\ref{Fig:PDF_nT}, the curves for the approximate PDF expression \eqref{Eq:Approximate_PDF_g} for $n_{\rm T}\geq2$ exhibit a good fit with the empirical ones. This fit improves as $\rho$ and/or $n_{\rm T}$ increase (or equivalently as the histogram becomes more heavy tailed). The latter substantiates that Theorem~\ref{Theorem:Approximate_Distribution_y} provides an accurate approximation for the PDF of ${\rm g}$ for multi-antenna communication networks with increased spatial fading correlation and very large numbers of antennas (massive MIMO).

Since $|g_i|^2$'s are independent and identically distributed RVs and independent from the PPP $\Phi$, the LT that appears in \eqref{Eq:Coverage_3} can be obtained using similar steps to \cite[Theorem 1]{J:Andrews_PPP_Cellular} as  
\begin{equation}\label{Eq:LT}
\mathcal{L}_{\Psi_0}\left(s\right) = \exp\left[-2\pi\lambda\int_0^\infty \mathcal{H}(s,x)f_{\rm g}(x){\rm d}x\right]
\end{equation}
where function $\mathcal{H}(s,x)$ is defined as
\begin{equation}\label{Eq:Function_H}
\begin{split}
\mathcal{H}(s,x) &\triangleq \int_0^\infty\left[1-\exp\left(-sxy^{-\alpha}\right)\right]y{\rm d}y
\\&\stackrel{(\rm{i})}{=} \frac{\left(sx\right)^{2/\alpha}}{2}\Gamma\left(1-\frac{2}{\alpha}\right).
\end{split}
\end{equation}
In \eqref{Eq:Function_H}, $(\rm{i})$ follows from the change of variables $y^{-\alpha}\rightarrow u$ and after some algebraic and calculus manipulations. Substituting \eqref{Eq:Function_H} into \eqref{Eq:LT}, then replacing either \eqref{Eq:Exact_PDF_g} or \eqref{Eq:Approximate_PDF_g}, and finally using the definition of the fractional moment, a closed-form expression for the LT of $\Psi_0$ can be obtained as
\begin{equation}\label{Eq:LT_1}
\mathcal{L}_{\Psi_0}\left(s\right) = \exp\left(-\xi\lambda s^{2/\alpha}\right)
\end{equation}
where the parameter $\xi\geq0$ is given by
\begin{equation}\label{Eq:xi}
\xi = \pi\Gamma\left(1-\frac{2}{a}\right)\mathbb{E}\{{\rm g}^{2/\alpha}\}.
\end{equation}
A closed-form exact expression for the special case of $n_{\rm T}=2$ and a closed-form analytical approximation for the general case $n_{\rm T}\geq2$ for the fractional moment $\mathbb{E}\{{\rm g}^{2/\alpha}\}$ that appears in \eqref{Eq:xi} are presented in the following two corollaries.

\newtheorem{corollary}{Corollary}
\begin{corollary}\label{Corollary:Exact_Moment}\rm
A closed-form exact expression for $\mathbb{E}\{{\rm g}^{2/\alpha}\}$ for the special case where $n_{\rm T}=2$ is given by    
\begin{equation}\label{Eq:Exact_Moment_2VSa_y}
\mathbb{E}\{{\rm g}^{2/\alpha}\} = \frac{1-\rho^2}{2\rho}\left[\mathcal{G}\left(1+\rho,\alpha\right)-\mathcal{G}\left(1-\rho,\alpha\right)\right]
\end{equation}
where function $\mathcal{G}(\eta,a)$ is defined as 
\begin{equation}\label{Eq:Function_G}
\begin{split}
\mathcal{G}(\eta,\alpha) =& \frac{\eta^{2/\alpha+1}}{4}\left\{\Gamma\left(\frac{2}{\alpha}+1\right)\left[\frac{\eta}{2-\eta}+\frac{\alpha}{\alpha+2}\right.\right.
\\&\times\left._2F_1\left(1,\frac{2}{\alpha}+1;\frac{2}{\alpha}+2;\frac{\eta}{2}\right)\right]-\frac{\alpha\eta}{4(\alpha+1)}
\\&\times\left.\Gamma\left(\frac{2}{\alpha}+2\right)\,_2F_1\left(1,\frac{2}{\alpha}+2;\frac{2}{\alpha}+3;\frac{\eta}{2}\right)\right\}
\end{split}
\end{equation}
\end{corollary}
\begin{proof}
Starting from the definition of the expectation operator and using \eqref{Eq:Exact_PDF_g} and \eqref{Eq:Function_F}, \eqref{Eq:Function_G} can be easily obtained with the use of \cite[eq. (3.381/4)]{B:Gra_Ryz_Book}, \cite[eq. (6.228/2)]{B:Gra_Ryz_Book} and after some elementary algebraic manipulations. 
\end{proof}

\begin{corollary}\label{Corollary:Aproximate_Moment}\rm
A closed-form approximate expression for $\mathbb{E}\{{\rm g}^{2/\alpha}\}$ for the case where $n_{\rm T}\geq2$ is obtained as 
\begin{equation}\label{Eq:Approximate_Moment_2VSa_y}
\mathbb{E}\{{\rm g}^{2/\alpha}\} = \sigma_{\rm g}^{2/\alpha}\Gamma\left(\frac{2}{\alpha}+1\right).
\end{equation}
\end{corollary}
\begin{proof}
Substituting \eqref{Eq:Approximate_PDF_g} and \eqref{Eq:sigma_g} into the definition of the expectation operator and making use of \cite[eq. (3.381/4)]{B:Gra_Ryz_Book} completes the proof.
\end{proof}
\textit{Remark 2:} Similar to the remark for Theorem~\ref{Theorem:Approximate_Distribution_y}, for the special case where $\mathbf{R}=\mathbf{I}_n$ $\forall$ $n=2,3,\ldots,n_{\rm T}$, Corollary~\ref{Corollary:Aproximate_Moment} presents an exact expression for $\mathbb{E}\{{\rm g}^{2/\alpha}\}$ by setting $\sigma_{\rm g}=1$.

To derive a closed-form exact expression for $P_{\rm suc}$ for the special case $n_{\rm T}=2$, \eqref{Eq:LT_1} is substituted into \eqref{Eq:Coverage_3} yielding
\begin{equation}\label{Eq:Coverage_nT2}
\begin{split}
P_{\rm suc} = \sum_{\ell=1}^2&\frac{1-(-1)^\ell\rho}{2\rho}\exp\left\{-\frac{r_{0,0}^{\alpha}\sigma^2\gamma_{\rm th}}{{\rm P}\left[1-(-1)^\ell\rho\right]}\right\}
\\&\times\exp\left\{-\xi\lambda r_{0,0}^2\left[\frac{\gamma_{\rm th}}{1-(-1)^\ell\rho}\right]^{2/\alpha}\right\}
\end{split}
\end{equation}
where $\xi$ is given by \eqref{Eq:xi} after substituting \eqref{Eq:Exact_Moment_2VSa_y} with \eqref{Eq:Function_G}. Similarly, substituting \eqref{Eq:LT_1} using \eqref{Eq:xi} and \eqref{Eq:Approximate_Moment_2VSa_y} into \eqref{Eq:Coverage_3}, the following closed-form approximation for $P_{\rm suc}$ is obtained
\begin{equation}\label{Eq:Coverage_nT}
\begin{split}
&P_{\rm suc} = 1-\sum_{\ell=1}^{n_{\rm T}} \prod_{\substack{i=1\\ i \neq \ell}}^{n_{\rm T}}\left(1-\frac{\left[\boldsymbol{\mu}\right]_i}{\left[\boldsymbol{\mu}\right]_\ell} \right)^{-1}
\\&\times\left\{1-\exp\left(-\frac{r_{0,0}^{\alpha}\sigma^2\gamma_{\rm th}}{{\rm P}\left[\boldsymbol{\mu}\right]_\ell}\right)\exp\left[-\xi\lambda r_{0,0}^2\left(\frac{\gamma_{\rm th}}{\left[\boldsymbol{\mu}\right]_\ell}\right)^{2/\alpha}\right]\right\}.
\end{split}
\end{equation}

\subsection{Area Spectral Efficiency}
Using the previously derived expressions for $P_{\rm suc}$, the area spectral efficiency of the considered downlink wireless network is straightforwardly obtained as
\begin{equation}\label{Eq:Area_Spectral_Efficiency}
{\rm ASE} = \lambda\log_2\left(1+\gamma_{\rm th}\right)P_{\rm suc}.
\end{equation}
By substituting \eqref{Eq:Coverage_nT2} for $n_{\rm T}=2$ and \eqref{Eq:Coverage_nT} for $n_{\rm T}\geq2$ in \eqref{Eq:Area_Spectral_Efficiency}, a closed-form exact and a closed-form approximate expression, respectively, for ${\rm ASE}$ can be easily derived.

\section{Performance Evaluation Results}\label{sec:Results}
In this section, numerically evaluated results for the performance expressions provided in Section~\ref{sec:Analysis} are presented and compared with equivalent results obtained by means of computer simulations. In particular, we numerically evaluate both the exact and the approximate expressions for $P_{\rm suc}$ given by \eqref{Eq:Coverage_nT2} and \eqref{Eq:Coverage_nT}, respectively, as well as those for ${\rm ASE}$ obtained from \eqref{Eq:Area_Spectral_Efficiency}. Without loss of generality, in the performance evaluation results that follow, exponential correlation matrices have been considered and it has been assumed that $\alpha=3.5$, $r_{0,0}=1$ and $\sigma^2/{\rm P}=1$.
\begin{figure}[!t]
\centering
\includegraphics[keepaspectratio,width=3.1in]{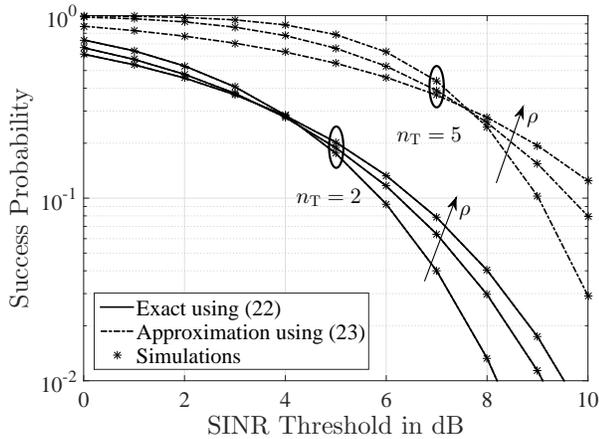}
\caption{Success probability, $P_{\rm suc}$, as a function of the SINR threshold $\gamma_{\rm th}$ in dB for $\lambda=10^{-4}$, $n_{\rm T}=2$ and $5$, as well as for $\rho=\{0.01,0.7,0.95\}$.}
\label{Fig:P_success}
\end{figure}

\begin{figure}[!t]
\centering
\includegraphics[keepaspectratio,width=3.1in]{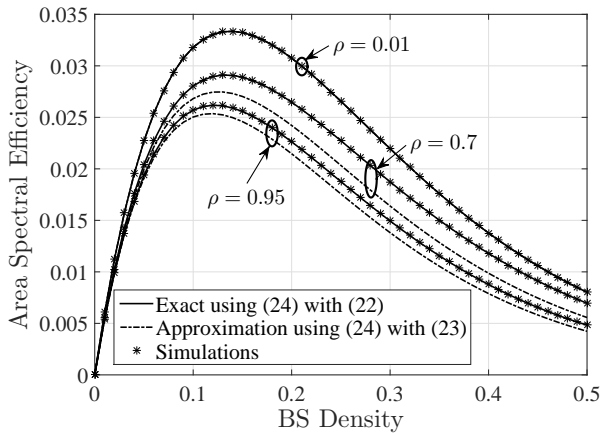}
\caption{Area spectral efficiency, ${\rm ASE}$, as a function of the BS density $\lambda$ for $\gamma_{\rm th}=3$ dB, $n_{\rm T}=2$, and $\rho=\{0.01,0.7,0.95\}$.}
\label{Fig:S_e_nT2}
\end{figure}

In Fig$.$~\ref{Fig:P_success}, $P_{\rm suc}$ is plotted versus $\gamma_{\rm th}$ for $\lambda=10^{-4}$, $n_{\rm T}=2$ and $5$, and for $\rho=\{0.01,0.7,0.95\}$. As shown for all cases, the numerically evaluated results for the derived exact ($n_{\rm T}=2$) and the approximate ($n_{\rm T}=5$) $P_{\rm suc}$ match perfectly and are sufficiently close, respectively, with equivalent simulations. In addition, it is demonstrated that $P_{\rm suc}$ improves with decreasing $\gamma_{\rm th}$ and increasing $n_{\rm T}$. Interestingly, there exists a critical $\gamma_{\rm th}$ value, $\gamma_{\rm th}^*$, (crossing point in the figure) that increases with increasing $n_{\rm T}$ that determines the impact of $\rho$ on $P_{\rm suc}$. More specifically, for $\gamma_{\rm th}\leq\gamma_{\rm th}^*$, $P_{\rm suc}$ improves with decreasing  $\rho$. However, when $\gamma_{\rm th}>\gamma_{\rm th}^*$,  as $\rho$ decreases $P_{\rm suc}$ degrades.    

By setting $\gamma_{\rm th}=3$ dB in Fig$.$~\ref{Fig:S_e_nT2}, ${\rm ASE}$ is illustrated as a function of $\lambda$ for $n_{\rm T}=2$ and $\rho=\{0.01,0.7,0.95\}$. As shown in this figure, there exists a perfect match between the evaluated results for the derived exact ${\rm ASE}$ and equivalent simulations. In addition, the derived approximation appears to be a tight lower bound, which tightness increases as $\rho$ tends to either of its two extreme values (i$.$e$.$ $\rho\rightarrow0$ or $\rho\rightarrow1$) as well as when $\lambda$ decreases. It is also demonstrated that ${\rm ASE}$ degrades with increasing $\rho$ and its maximum value happens for lower values of $\lambda$ as $\rho$ increases. The same performance metric over $\lambda$ is illustrated in Fig$.$~\ref{Fig:S_e_nT5-20} for cases where $\rho=0.7$ and $0.95$, 
$n_{\rm T}$ takes the values $5$ and $20$, and $\gamma_{\rm th}$ is increased to $5$ dB. In this figure, it is shown that ${\rm ASE}$ improves with increasing $n_{\rm T}$. In contrast to Fig$.$~\ref{Fig:S_e_nT2}, as $\rho$ increases, ${\rm ASE}$ is shown to improve. Moreover, the maximum value for ${\rm ASE}$ happens for lower values of $\lambda$ as $\rho$ decreases. As also discussed in the description of Fig$.$~\ref{Fig:P_success}, the latter behavior implies that, for the parameters' setting in Figs$.$~\ref{Fig:S_e_nT2} and \ref{Fig:S_e_nT5-20}, $\rho$ has different impact on ${\rm ASE}$. Finally, it is noted that, as demonstrated in Fig$.$~\ref{Fig:S_e_nT5-20}, the impact of $\rho$ on ${\rm ASE}$ increases with increasing $n_{\rm T}$.
\begin{figure}[!t]
\centering
\includegraphics[keepaspectratio,width=3.1in]{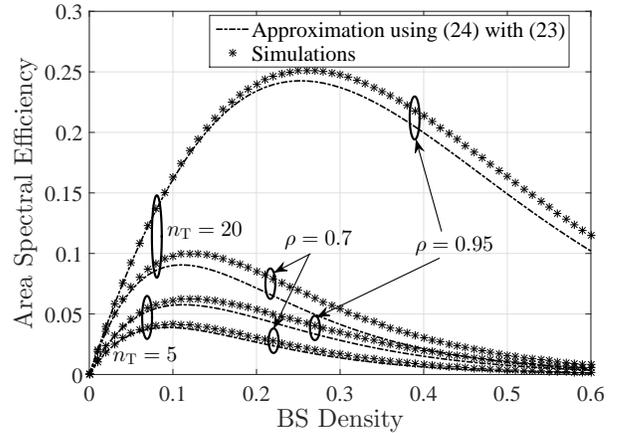}
\caption{Area spectral efficiency, ${\rm ASE}$, as a function of the BS density $\lambda$ for $\gamma_{\rm th}=5$ dB, $n_{\rm T}=5$ and $20$, as well as for $\rho=0.7$ and $0.95$.}
\label{Fig:S_e_nT5-20}
\end{figure}

\section{Conclusion}\label{sec:Conclusion}
In this paper, we have analyzed the performance of MRT in a wireless Poisson network assuming spatially arbitrarily correlated Rayleigh fading channels. Novel closed-form expressions for the distribution of the power of the interference resulting from the coexistence of one intended and one unintended MRT over the considered fading channels were first presented. Based on these expressions and using tools from stochastic geometry, we further derived closed-form expressions for the success probability and area spectral efficiency of the downlink network under investigation. Our theoretical analysis, which is corroborated with simulation results, shed light on how spatial fading correlation affects the network performance depending on the interplay between the density of BSs and the target SINR threshold. Interestingly, depending on the background noise level, spatial fading correlation may increase or decrease the network throughput, having a phase transition for a certain SINR threshold value.

\appendices
\renewcommand{\theequation}{A.\arabic{equation}}
\setcounter{equation}{0}
\section{Proof of Theorem~\ref{Theorem:Exact_Distribution_y}}\label{App:Proof_Theorem_1}
Starting from \eqref{Eq:g_variable} and using the definition of the channel gain vectors $\mathbf{\tilde{h}}_{0,i}$ and $\mathbf{\tilde{h}}_{i,i}$ in Section~\ref{sec:System_Model}, $g_i$ $\forall$ $i\in\mathcal{N}_{+}$ can be equivalently expressed as $g_i=\mathbf{h}_{0,i}^{\rm H}\mathbf{w}_{i}$, with the random vector $\mathbf{w}_i\in\mathcal{C}^{2}$ defined as
\begin{equation}\label{Eq:w}
\mathbf{w}_i = \left(\mathbf{R}_{0,i}^{1/2}\right)^{\rm H}\mathbf{R}_{i,i}^{1/2}\frac{\mathbf{h}_{i,i}}{\|\mathbf{R}_{i,i}^{1/2}\mathbf{h}_{i,i}\|}.
\end{equation}
It is obvious that $\overline{g}_i=\mathbf{w}_{i}^{\rm H}\mathbf{h}_{0,i}$ $\forall$ $i\in\mathcal{N}_{+}$ conditioned on $\mathbf{h}_{i,i}$ is a linear functional of the standard circularly-symmetric jointly-Gaussian random vector $\mathbf{h}_{0,i}$, and hence it is a circularly-symmetric Gaussian RV \cite[Sec. 3.7]{B:Gallager}. The mean value of each $\overline{g}_i$ can be easily verified to be independent of $\mathbf{h}_{i,i}$ and given by 
\begin{equation}\label{Eq:mean_g}
\psi = \sum_{n=1}^{n_{\rm T}} \mathbb{E}\left\{\left[\mathbf{h}_{0,i}\right]_n\right\}\left[\mathbf{w}_i^{\rm H}\right]_n \stackrel{(b)}{=} 0
\end{equation}
where $(b)$ follows from our assumption on the expectation of the elements of $\mathbf{h}_{0,i}$ in Section~\ref{sec:System_Model}. In addition, it can be shown that each $\left|g_i\right|^2$ conditioned on $\mathbf{h}_{i,i}$ is exponentially distributed with scale parameter obtained as 
\begin{equation}\label{Eq:conditional_variance_g}
\mathbb{E}\left\{\left|g_i\right|^2\left.\right|\mathbf{h}_{i,i}\right\} = \frac{\mathbf{h}_{i,i}^{\rm H}\mathbf{R}^2\mathbf{h}_{i,i}}{\mathbf{h}_{i,i}^{\rm H}\mathbf{R}\mathbf{h}_{i,i}}.
\end{equation}
In the derivation of the latter expression, we have used the notation $\mathbf{R}^2\triangleq\mathbf{R}\mathbf{R}$ and the fact that $\mathbb{E}\{\mathbf{h}_{0,i}\mathbf{h}_{0,i}^{\rm H}\}=\mathbf{I}_2$. 

To obtain an analytical expression for the distribution of $\left|g_i\right|^2$ $\forall$ $i\in\mathcal{N}_{+}$, we need first to derive an analytical expression for the distribution of the generalized Rayleigh quotient that appears in the right-hand side of \eqref{Eq:conditional_variance_g}, i$.$e$.$ for the positive real-valued RV $r_i=\mathbf{h}_{i,i}^{\rm H}\mathbf{R}^2\mathbf{h}_{i,i}/\mathbf{h}_{i,i}^{\rm H}\mathbf{R}\mathbf{h}_{i,i}$. Using the results of \cite{J:Alexandg2015}, a closed-form expression for the PDF of each $r_i$ can be obtained as 
\begin{equation}\label{Eq:PDF_r}
f_{r_i}(y) =
\left\{
\begin{array}{ll}
\frac{1-\rho^2}{2\rho}\left(2-y\right)^{-2},&1-\rho<y<1+\rho\\
0,&{\textrm{otherwise}}
\end{array}
\right..
\end{equation}
It is noted that, as shown in \eqref{Eq:PDF_r}, the PDF of each $r_i$ depends only on $\rho$, and hence is identical for all $r_i$'s (by assumption, $\mathbf{R}$ is common for $\mathbf{\tilde{h}}_{i,i}$ $\forall$ $i\in\mathcal{N}_{+}$). By using the exponential conditional PDF of $\left|g_i\right|^2$ and \eqref{Eq:PDF_r}, the unconditional PDF of each $\left|g_i\right|^2$ can be finally computed as
\begin{equation}\label{Eq:Exact_PDF_y_1}
\begin{split}
f_{\rm g}(x) &= \int_{0}^{\infty}\frac{1}{y}\exp\left(-\frac{x}{y}\right)f_{r_i}(y){\rm d}y
\\&=\frac{1-\rho^2}{2\rho}\int_{1-\rho}^{1+\rho}\frac{\exp\left(-\frac{x}{y}\right)}{y\left(2-y\right)^2}{\rm d}y,
\end{split}
\end{equation}
which using \cite[eqs. (3.352)]{B:Gra_Ryz_Book} and after some algebraic manipu\-lations yields \eqref{Eq:Exact_PDF_g}. 

\renewcommand{\theequation}{B.\arabic{equation}}
\setcounter{equation}{0}
\section{Proof of Theorem~\ref{Theorem:Approximate_Distribution_y}}\label{App:Proof_Theorem_2}
Both $\mathbf{R}^2$ and $\mathbf{R}$ that appear in the scale parameter described by \eqref{Eq:conditional_variance_g} are real symmetric matrices, and hence their eigenvalue decomposition is computed as $\mathbf{R}^2=\mathbf{V}{\rm diag}\left\{\boldsymbol{\nu}\right\}\mathbf{V}^{\rm H}$ and $\mathbf{R}=\mathbf{V}{\rm diag}\left\{\boldsymbol{\mu}\right\}\mathbf{V}^{\rm H}$, respectively, with $\mathbf{V}\in\mathcal{C}^{n_{\rm T}\times n_{\rm T}}$ being an orthogonal matrix and $\left[\boldsymbol{\nu}\right]_n=\left[\boldsymbol{\mu}\right]_n^2$ $\forall$ $n=2,3,\dots,n_{\rm T}$. By introducing the random vector $\mathbf{q}_i=\mathbf{V}^{\rm H}\mathbf{h}_{i,i}$, which can be ea\-sily shown to be circularly-symmetric jointly-Gaussian, \eqref{Eq:conditional_variance_g} can be re-expressed as $\mathbb{E}\{\left|g_i\right|^2\left.\right|\mathbf{h}_{i,i}\}=\mathbb{E}\{\left|g_i\right|^2\left.\right|r_i\}=r_i$, where the positive real-valued RV $r_i$ that was introduced in the previous appendix is hereinafter equivalently given by
\begin{equation}\label{Eq:r}
r_i = \sum_{n=1}^{n_{\rm T}}\frac{\left[\boldsymbol{\mu}\right]_n^2\left|\left[\mathbf{q}_i\right]_n\right|^2}{\left[\boldsymbol{\mu}\right]_n\left|\left[\mathbf{q}_i\right]_n\right|^2}.
\end{equation}

To derive an approximate expression for the PDF of each $\left|g_i\right|^2$, we make use of one of the results presented in Appendix~\ref{App:Proof_Theorem_1} according to which, the PDF of each $\left|g_i\right|^2$ is an exponential one with scale parameter depending on $\mathbf{h}_{i,i}$, or equivalently on $r_i$ as expressed as \eqref{Eq:r}. We propose the approximation of the PDF of $\left|g_i\right|^2$ $\forall$ $i\in\mathcal{N}_{+}$ with an exponential PDF having as a scale parameter the $\mathbb{E}\{\left|g_i\right|^2\}=\mathbb{E}\{r_i\}$, which does not depend on $r_i$. To this end, $\mathbb{E}\{\left|g_i\right|^2\}$ can be approximated as follows
\begin{equation}\label{Eq:Expectation_Approximations}
\begin{split}
\mathbb{E}\{\left|g_i\right|^2\}&\cong\mathbb{E}\left\{\sum_{n=1}^{n_{\rm T}}\left[\boldsymbol{\mu}\right]_n^2\omega_{i,n}\right\}  \mathbb{E}\left\{\left(\sum_{n=1}^{n_{\rm T}}\left[\boldsymbol{\mu}\right]_n\omega_{i,n}\right)^{-1}\right\}
\\&\cong\mathbb{E}\left\{\sum_{n=1}^{n_{\rm T}}\left[\boldsymbol{\mu}\right]_n^2\omega_{i,n}\right\}/\mathbb{E}\left\{\sum_{n=1}^{n_{\rm T}}\left[\boldsymbol{\mu}\right]_n\omega_{i,n}\right\}
\end{split}
\end{equation}
where $\omega_{i,n}=\left|\left[\mathbf{q}_i\right]_n\right|^2$. By using the fact that $\mathbb{E}\{\omega_{i,n}\}=1$ $\forall$ $i\in\mathcal{N}_{+}$ and $\forall$ $n=2,3,\dots,n_{\rm T}$, the right-hand side of \eqref{Eq:Expectation_Approximations} can be expressed $\forall$ $i\in\mathcal{N}_{+}$ as in \eqref{Eq:sigma_g}.


\bibliographystyle{IEEEtran}
\bibliography{IEEEabrv,references}

\end{document}